\documentclass[11pt]{article}

\usepackage{mathpazo}
\usepackage{amsfonts}
\usepackage{amsmath}
\usepackage{graphicx}
\usepackage{latexsym}
\usepackage{mathabx}
\usepackage{bm}

\usepackage[margin=1.05in]{geometry}
  
\usepackage[T1]{fontenc}
\usepackage{times}
\usepackage{color,graphicx}
\usepackage{array}
\usepackage{enumerate}
\usepackage{amsmath}
\usepackage{amssymb}
\usepackage{amsthm}
\usepackage{pgfplots}
\usepackage{pgf}
\usepackage{tikz}
\usetikzlibrary{patterns}
\usepgfplotslibrary{patchplots} 
\usetikzlibrary{decorations.markings}
\usetikzlibrary{pgfplots.patchplots} 
\pgfplotsset{width=9cm,compat=1.12}
    \usepgfplotslibrary{colormaps,external}
    \definecolor{ocre}{RGB}{0,96,128}
    \definecolor{purp}{RGB}{112,0,112}
    \pgfplotsset{
    	colormap={ocrefade}{
    		rgb255=(150,216,255)
    		rgb255=(0,64,96)
    	}
    }

\usepackage{amsthm}

\newtheorem{definition}{Definition}
\newtheorem{lemma}{Lemma}
\newtheorem{theorem}{Theorem}
\newtheorem{corollary}{Corollary}

\newtheorem{example}{Example} 
\DeclareMathOperator*{\argmin}{argmin}

\usepackage{xcolor}
\usepackage{makeidx}
\usepackage[colorlinks=true,linkcolor=blue,anchorcolor=blue,citecolor=red,urlcolor=magenta]{hyperref}

\newcommand{\tr}{\operatorname{Tr}}
\newcommand{\rank}{\operatorname{rank}}

\newcommand{\bra}[1]{\langle #1 |}
\newcommand{\ket}[1]{| #1 \rangle}

\newcommand{\ketbra}[2]{| #1 \rangle\langle #2 |}

\newcommand{\defeq}{\stackrel{\smash{\textnormal{\tiny def}}}{=}}

\begin{document}
\title{Absolutely $k$-Incoherent Quantum States and \\ Spectral Inequalities for Factor Width of a Matrix}

\author{
	Nathaniel Johnston,\textsuperscript{\!\!1,2} \ Shirin Moein,\textsuperscript{\!\!1,2,3,4} \ Rajesh Pereira,\textsuperscript{\!\!2} \ and Sarah Plosker\textsuperscript{2,3}
}

\maketitle

\begin{abstract}
	We investigate the set of quantum states that can be shown to be $k$-incoherent based only on their eigenvalues (equivalently, we explore which Hermitian matrices can be shown to have small factor width based only on their eigenvalues). In analogy with the absolute separability problem in quantum resource theory, we call these states ``absolutely $k$-incoherent'', and we derive several necessary and sufficient conditions for membership in this set. We obtain many of our results by making use of recent results concerning hyperbolicity cones associated with elementary symmetric polynomials.\\
	
	

\end{abstract}

\addtocounter{footnote}{1}
\footnotetext{Department of Mathematics \& Computer Science, Mount Allison University, Sackville, NB, Canada E4L 1E4}\addtocounter{footnote}{1}
\footnotetext{Department of Mathematics \& Statistics, University of Guelph, Guelph, ON, Canada N1G 2W1} \addtocounter{footnote}{1}
\footnotetext{Department of Mathematics \& Computer Science, Brandon University, Brandon,
    MB, Canada R7A 6A9}\addtocounter{footnote}{1}
\footnotetext{Department of Mathematical Sciences, Isfahan University of Technology, Isfahan, Iran 84156-83111}

\section{Introduction}

In the theory of quantum information, there are numerous resource theories that can be used to make rigorous the idea of certain quantum states being ``useful'' or ``useless'' \cite{CG18}. Perhaps the most well-known of these resource theories is that of entanglement and separability \cite{Vid00}, where separable states are those that are ``useless'', while entangled (i.e., non-separable) states are those that are ``useful'' (in terms of quantum channel discrimination, for example \cite{PW09}).

The well-studied \emph{absolute} separability problem asks for a characterization of which quantum states can be determined to be separable based only on their spectrum (i.e., multiset of eigenvalues) \cite{KZ00}. A full solution to this problem remains out of reach \cite{AJR15}, but a complete characterization is known in small dimensions \cite{VAD01,Joh13} and some non-trivial necessary conditions are known in all dimensions \cite{Hil07}. This ``absolute'' question has been asked for some other resource theories and in some other contexts as well. For example, it has been explored in the resource theory of \emph{symmetric} separability \cite{CJMP21,SM21}, and it has been explored in the context of the reduction map from quantum information theory \cite{JLN15}.

In this work, we introduce the corresponding ``absolute'' question for another resource theory: the resource theory of $k$-coherence \cite{RBC18}, which has been studied by several groups in many different contexts \cite{4,9,10,11}. Understanding the structure of $k$-coherence would be useful in many fields, as it helps explain the function of complex biological molecules like those found in light-harvesting \cite{5} and chemical and biophysical systems \cite{6,7}, for example, and it can be used to describe the statistical properties of a quantum state's interference pattern \cite{8}.

\subsection{Measuring k-Incoherence in a Laboratory Setting}

Despite widespread interest in quantifying coherence, there is a lack of efficient methods for measuring coherence in experiments, which can make it difficult to use in practice. One of the standard ways of using measurement results to show that a quantum state $\rho$ is $k$-coherent is to measure it against a $k$-coherence witness $W$: if $\tr(W\rho) < 0$ then $\rho$ must be $k$-coherent. In fact, this method can be used to obtain lower bounds on $\rho$'s robustness of $k$-coherence (see \cite{NBCPJA16,WSR21}, for example).

In the other direction, however, it seems to be much more difficult to use measurement results to show that a quantum state is $k$-incoherent. This problem was partially overcome in \cite{YG19} by developing a method of estimating the $k$-coherence in a quantum state via its spectrum. Our work on the ``absolute'' version of $k$-coherence can be seen as an extension of this idea---we ask which quantum states can be shown to be $k$-incoherent (i.e., ``useless'' in this resource theory) based only on knowledge of their spectrum.

More specifically, a quantum state acting on an $n$-dimensional space requires $n^2 - 1$ real parameters to specify, and thus $n^2 - 1$ measurement outcomes to reconstruct via tomography. However, it is often must easier to determine just the state's eigenvalues, rather than the structure of the entire state \cite{EAO02,TOKKNN13, teo2021modern}. After all, the eigenvalues are just $n-1$ real parameters instead of $n^2 - 1$. Our results thus provide tests that can be used to show that quantum states are $k$-incoherent, just given this (easier-to-obtain) restricted information about the state's eigenvalues: if the eigenvalues satisfy any of our sufficient conditions for absolute $k$-incoherence, then the corresponding quantum state must be $k$-incoherent.

In another direction, our results shed light on how much noise must be added to the maximally mixed state in order for it to become $k$-coherent. Indeed, every state that is sufficiently close to the maximally mixed state (which has all eigenvalues equal to each other) is $k$-incoherent, and our sufficient conditions for absolute $k$-incoherence give quantitative statements about how spread out the eigenvalues have to be before $k$-coherence is possible. Phrased another way, our results show that if the (spectral, Frobenius, trace, or any other unitarily-invariant) norm distance between a quantum state and the maximally mixed state is small, then that state must be $k$-incoherent, meaning that $k$-incoherence in this regime is resistant to errors is measurement results. We note that results of this type have been studied extensively in the resource theory of entanglement \cite{GB02}, but to our knowledge no substantial results in this direction were known for $k$-coherence prior to ours.

\subsection{Summary of our Results and Methods}

Much like known results for the absolute separability problem, we obtain a complete characterization of absolutely $k$-incoherent states in small dimensions, as well as non-trivial one-sided (i.e., necessary or sufficient, but not both) conditions in all dimensions. We also obtain a complete necessary and sufficient criterion for the $k = n-1$ case (in all dimensions), showing that membership in this set can be determined in polynomial time. In the terminology of pure mathematics, $k$-incoherent states are exactly the matrices that have factor width at most $k$ \cite{boman2005factor}, and our results can thus equivalently be interpreted as spectral inequalities that can be used to show that a matrix has small factor width.

To arrive at our results, we consider a family of matrices that we call \emph{$k$-locally PSD} (in keeping with the terminology of \cite{blekherman2022hyperbolic}, where these matrices were introduced independently). These matrices form exactly the dual cone of the set of $k$-incoherent quantum states, so they serve as ``witnesses'' in the resource theory of $k$-coherence (in the exact same sense that entanglement witnesses \cite{Ter00} can be used to ``witness'' entanglement). In particular, we show that the spectra of these $k$-locally PSD matrices belong to certain sets called ``hyperbolicity cones'' \cite{R06,Z8}, and we use recent results about the dual of these cones to get bounds on the spectrum of absolutely $k$-incoherent states. We also provide a tight characterization of these cones when $k = n-1$, and recover (in a slightly more explicit way) a result from \cite{blekherman2022hyperbolic}.

Our paper is organized as follows: In Section~\ref{sec:preliminaries}, we review the various definitions and concepts from quantum information theory that we will be exploring, including  $k$-incoherent states and their connection to hermitian matrices with factor width at most $k$, and we introduce the notion of $k$-locally PSD matrices via dual cones. In Section~\ref{sec:factor_positive_spectra} we develop some bounds on the possible spectra of $k$-locally PSD matrices, in Section~\ref{sec:computational_method_factor_pos} we present a method for numerically constructing a $k$-locally PSD matrix with a given spectrum, and then in Section~\ref{sec:hyperbolic_cones} we discuss how our results about $k$-locally PSD matrices relate to recent work on hyperbolicity cones. Readers who are just interested in our results on absolute $k$-incoherence can skip Sections~\ref{sec:computational_method_factor_pos} and~\ref{sec:hyperbolic_cones}, and perhaps all of Section~\ref{sec:factor_positive_spectra}.

In Section~\ref{sec:abs_kcoh}, we present out main results about absolute $k$-incoherence, which include a complete characterization of absolute $2$-incoherence in dimension $3$ (Theorem~\ref{thm:abs_2incoh}), a complete characterization of absolute $(n-1)$-incoherence in every dimension (Theorem~\ref{thm:abs_n_min1_incoherence}), some sufficient conditions for absolute $k$-incoherence for all values of $k$ in all dimensions (Theorems~\ref{thm:dual_ck_abs_coh} and~\ref{thm:k_incoh_max_eig}), as well as a corresponding necessary condition (Theorem~\ref{thm:k_incoh_projection}). Finally, we close in Section~\ref{sec:conclusions} with some open questions.

\section{Mathematical and Quantum Information Theory Preliminaries}\label{sec:preliminaries}

The notation and terminology that we use is fairly standard, so we just introduce it briefly. For a more thorough introduction to quantum information theory, see any of numerous standard textbooks like \cite{NC00,Wat18}.

We use $M_n$ to denote the set of all $n\times n$ matrices with complex entries, $M_n^{\textup{H}}$ to denote the set of Hermitian matrices (i.e., matrices $A \in M_n$ having $A^* = A$, where $A^*$ is the conjugate transpose of $A$), and $M_n^+$ for the subset of them that are (Hermitian) positive semidefinite. We use bold lower case letters such as $\mathbf{v}$ and $\mathbf{w}$ to denote vectors in $\mathbb{C}^n$, and $v_1$, $v_2$, $\ldots$, $v_n$ to denote the entries of $\mathbf{v}\in \mathbb C^n$.

In quantum information theory, a \textbf{mixed quantum state} or \textbf{density matrix} is a positive semidefinite trace-one matrix $\rho \in M_n^{+}$. A \textbf{pure quantum state} is a unit vector in $\mathbb{C}^n$, which we denote using ``bra-ket'' notation: $\ket{v} \in \mathbb{C}^n$ is a unit column vector, while $\bra{v} \defeq \ket{v}^*$ is the corresponding (dual) row vector. Whenever we use lowercase Greek letters like $\rho$, we are assuming that it is a quantum state normalized to have $\tr(\rho) = 1$ (i.e., trace one) and similarly if we use $\ket{v}$ then we are assuming it is a \emph{unit} vector. We denote matrices and vectors that are not necessarily normalized like $A \in M_n$ and $\mathbf{v} \in \mathbb{C}^n$, respectively.

\subsection{Incoherence and Factor Width}\label{sec:incoh_factor_width}

By the spectral decomposition, every mixed state $\rho \in M_n^{+}$ can be written in the form
\[
    \rho = \sum_j \mathbf{v_j}\mathbf{v}_{\mathbf{j}}^*
\]
for some $\{\mathbf{v_j}\} \subseteq \mathbb{C}^n$. In the special case when each $\mathbf{v_j}$ can be chosen to have at most $k$ non-zero entries (where $1 \leq k \leq n$), $\rho$ is called \textbf{$\mathbf{k}$-incoherent} \cite{RBC18}. If $\rho$ is \emph{not} $k$-incoherent then it is called \textbf{$\mathbf{k}$-coherent}. Density matrices (or more generally, positive semidefinite matrices) that are $k$-incoherent, but not $(k-1)$-incoherent, are sometimes said to have \textbf{factor width} $k$ \cite{boman2005factor}.

We denote the (closed and convex) set of $k$-incoherent quantum states in $M_n^{+}$ by $\mathcal{I}_{k,n}$, and we note that we have the chain of inclusions
\[
    \mathcal{I}_{1,n} \subsetneq \mathcal{I}_{2,n} \subsetneq \cdots \subsetneq \mathcal{I}_{n-1,n} \subsetneq \mathcal{I}_{n,n}.
\]
In the $k = 1$ case, $\mathcal{I}_{1,n}$ is exactly the set of diagonal density matrices, which are simply said to be \textbf{incoherent}. At the other extreme, if $k = n$ then $\mathcal{I}_{n,n}$ is the set of \emph{all} density matrices.

\subsection{Dual Cones and Local Positivity}\label{sec:dual_cones}

Let $\mathcal{C}$ be a subset of a finite-dimensional real inner product space $\mathcal{V}$. Then the dual cone of $\mathcal{C}$, denoted by $\mathcal{C}^\circ$, is defined as follows:
\[
    \mathcal{C}^{\circ} \defeq \big\{ \mathbf{w} \in \mathcal{V} : \langle \mathbf{v},\mathbf{w}\rangle\geq 0 \text{ for all } \mathbf{v} \in \mathcal{C}\big\}.
\]
For example, if $\mathcal{V} = M_n^\textup{H}$ (equipped with the usual Hilbert--Schmidt inner product $\langle X,Y \rangle := \tr(XY)$), the dual cone of a set $\mathcal{C} \subseteq M_n^{\textup{H}}$ is
\[
    \mathcal{C}^{\circ} \defeq \big\{ Y \in M_n^{\textup{H}} : \tr(XY) \geq 0 \text{ for all } X \in \mathcal{C}\big\}.
\]
Similarly, if $\mathcal{V} = \mathbb{R}^n$ (equipped with the usual dot product) then the dual cone of a set $\mathcal{C} \subseteq \mathbb{R}^n$ is
\[
    \mathcal{C}^{\circ} \defeq \big\{ \mathbf{w} \in \mathbb{R}^n : \mathbf{v}\cdot\mathbf{w} \geq 0 \text{ for all } \mathbf{v} \in \mathcal{C}\big\}.
\]
Regardless of $\mathcal{V}$ and the structure of $\mathcal{C}$, every dual cone $\mathcal{C}^\circ$ is (as the name suggests) a cone: if $\mathbf{w} \in \mathcal{C}^\circ$ then $\lambda \mathbf{w} \in \mathcal{C}^{\circ}$ for all $0 \leq \lambda \in \mathbb{R}$. Dual cones are also always closed and convex, and the double-dual of any cone is the closure of its convex hull: $\mathcal{C}^{\circ\circ} = \overline{\mathrm{conv}(\mathcal{C})}$ (see \cite{BV04}, for example).

We will be particularly interested in the dual cone $\mathcal{I}_{k,n}^\circ \subseteq M_n^{\textup{H}}$ of the set of $k$-incoherent density matrices, which we characterize in a few different ways:

\begin{definition}\label{defn:k_fac_pos}
    We say that a Hermitian matrix $X \in M_n^{\textup{H}}$ is \textbf{k-locally PSD} if it has any of the following equivalent properties:
    
    \begin{enumerate}
        \item[a)] $X \in \mathcal{I}_{k,n}^\circ$ (i.e., $\tr(X\rho) \geq 0$ for all $\rho \in \mathcal{I}_{k,n}$);
        
        \item[b)] $\bra{v}X\ket{v} \geq 0$ for all $\ket{v} \in \mathbb{C}^n$ with at most $k$ non-zero entries; or
        
        \item[c)] every $k \times k$ principal submatrix of $X$ is positive semidefinite.
    \end{enumerate}
\end{definition}

The terminology ``$k$-locally PSD'' is from \cite{blekherman2022hyperbolic}, where the eigenstructure of such matrices was studied. We will also explore the eigenstructure of these matrices shortly, in Section~\ref{sec:factor_positive_spectra}.

We now show that these three defining properties really are equivalent. To see that (a) implies (b), we prove the contrapositive: if (b) is false then there exists a complex unit vector $\ket{v}$ with at most $k$ non-zero entries such that $\bra{v}X\ket{v} < 0$. By choosing $\rho = \ketbra{v}{v}$ we then have $\rho \in \mathcal{I}_{k,n}$ and $\tr(X\rho)<0$, showing that (a) is also false. Conversely, to see that (b) implies (a), write $\rho=\mathbf{v_j}\mathbf{v}_{\mathbf{j}}^*$ where each $\mathbf{v_j}$ has at most $k$ non-zero entries. Then
$$\tr(X\rho)=\sum_j \tr\big(X(\mathbf{v_j}\mathbf{v}_{\mathbf{j}}^*)\big)=\sum_j \mathbf{v}_{\mathbf{j}}^*X\mathbf{v_j} \geq 0,$$ 
where we used (b) to see that each in the sum is non-negative.

Finally, to see that (b) is equivalent to (c), we note that requiring $\bra{v}X\ket{v} \geq 0$ for all $\ket{v} \in \mathbb{C}^n$ with $k$ non-zero entries in positions $i_1$, $i_2$, $\ldots$, $i_k$ is equivalent to requiring that the $k \times k$ principal submatrix of $X$ corresponding to rows and columns $i_1$, $i_2$, $\ldots$, $i_k$ is positive semidefinite. Letting the tuple $(i_1,i_2,\ldots,i_k)$ range over all possible $k$-tuples of indices gives the equivalence of (b) and (c).

Property~(a) of Definition~\ref{defn:k_fac_pos} justifies thinking of these $k$-locally PSD matrices as ``$k$-coherence witnesses'' (in direct analogy with entanglement witnesses \cite{Ter00}). Indeed, if $X$ is $k$-locally PSD and $\tr(X\rho) < 0$ then we know that $\rho$ must be $k$-coherent. Since the quantity $\tr(X\rho)$ is measurable in a lab, it lets us use $X$ to ``witness'' the $k$-coherence of $\rho$.

It is straightforward to see that these sets of $k$-locally PSD matrices satisfy the chain of inclusions
\[
    M_n^{+} = \mathcal{I}_{n,n}^\circ \subsetneq \mathcal{I}_{n-1,n}^\circ \subsetneq \cdots \subsetneq \mathcal{I}_{2,n}^\circ \subsetneq \mathcal{I}_{1,n}^\circ,
\]
where $\mathcal{I}_{1,n}^\circ$ is the set of Hermitian matrices with non-negative diagonal entries.

\section{Spectra of $k$-locally PSD Matrices}\label{sec:factor_positive_spectra}

In this section, we investigate some bounds on the possible spectra of $k$-locally PSD matrices. We start with a result that tells us exactly how many negative eigenvalues these matrices can have:

\begin{theorem}\label{thm:neg_eigs_of_kfacpos}
    Every matrix $X \in \mathcal{I}_{k,n}^\circ$ has at most $n-k$ negative eigenvalues. Furthermore, this bound is tight: for all $1 \leq k \leq n$ there exists $X \in \mathcal{I}_{k,n}^\circ$ with exactly $n-k$ negative eigenvalues.
\end{theorem}

\begin{proof}
    Suppose $X \in \mathcal{I}_{k,n}^\circ$ has eigenvalues $\lambda_1 \geq \lambda_2 \geq \cdots \geq \lambda_n$. To see that it has at most $n-k$ negative eigenvalues (i.e., $\lambda_k \geq 0$), recall that eigenvalue interlacing (see \cite[Theorem 1.1]{LPN12}, for example) tells us if $\mu_1 \geq \mu_2 \geq \cdots \geq \mu_k$ are the eigenvalues of a $k \times k$ principal submatrix of $X$ then
    \begin{align}\label{eq:interlacing}
        \lambda_j\geq \mu_j \geq \lambda_{n-k+j} \quad \text{for all} \quad 1 \leq j \leq k.
    \end{align}
    Since each $k \times k$ submatrix of $X$ is positive semidefinite, we know that $\mu_k \geq 0$, so taking $j = k$ in Inequality~\eqref{eq:interlacing} gives $\lambda_k \geq \mu_k \geq 0$, as desired.
    
    To see that there exists $X \in \mathcal{I}_{k,n}^\circ$ with exactly $n-k$ negative eigenvalues, recall from \cite[Lemmas~8 and~9]{CW08} that if $0 < v_1 < v_2 < \cdots < v_n \in \mathbb{R}$ and
    \begin{align}\label{eq:vandermonde}
        V & := \begin{bmatrix}
            1 & v_1 & v_1^2 & \cdots & v_1^{n-1} \\
            1 & v_2 & v_2^2 & \cdots & v_2^{n-1} \\
            \vdots & \vdots & \vdots & \ddots & \vdots \\
            1 & v_n & v_n^2 & \cdots & v_n^{n-1} \\
        \end{bmatrix}
    \end{align}
    is a Vandermonde matrix, then any non-zero linear combination of $n-k$ of $V$'s columns has at most $n-k-1$ entries equal to $0$ (i.e., at least $k+1$ non-zero entries). With this in mind, we let $\mathbf{v_j}$ denote the $j$-th column of $V$ and then define the subspace
    \[
        \mathcal{S}_{n-k} := \mathrm{span}\{ \mathbf{v_1}, \mathbf{v_2}, \ldots, \mathbf{v_{n-k}} \},
    \]
    which only contains vectors with at least $k+1$ non-zero entries.
    
    If $P_{n-k}$ is the orthogonal projection onto $\mathcal{S}_{n-k}$ then it is clear that $X := I - cP_{n-k}$ has exactly $n-k$ negative eigenvalues whenever $c > 1$ (since the $\mathbf{v_j}$'s are linearly independent, so $\mathrm{rank}(P_{n-k}) = n-k$). We furthermore claim that $X$ is $k$-locally PSD for some $c > 1$. To see this, notice that if $\ket{v} \in \mathbb{C}^n$ has at most $k$ non-zero entries then $\bra{v}P_{n-k}\ket{v} < 1$. Compactness of the set of all such $\ket{v}$ implies that there exists a scalar $d < 1$ such that $\bra{v}P_{n-k}\ket{v} \leq d$ for all such $\ket{v}$, so $\bra{v}X\ket{v} = 1 - c\bra{v}P_{n-k}\ket{v} \geq 1 - cd$, which is non-negative as long as $c \leq 1/d$. Since $1/d > 1$, we can choose $c = 1/d$ to complete the proof.
\end{proof}

For example, if $k = n-1$ then the construction given in the proof of Theorem~\ref{thm:neg_eigs_of_kfacpos} says that there exists $c > 1$ such that $X := I - cP_{1} = I - c(\mathbf{1}\mathbf{1}^T/n)$ is $(n-1)$-locally PSD with $1$ negative eigenvalue, where $\mathbf{1} \defeq (1,1,\ldots,1)^T$ is the all-ones vector. In this case, we can explicitly compute $c = n/(n-1)$ as the largest possible value of $c$ that results in $X$ being $(n-1)$-locally PSD.


We now investigate the question of \emph{how} negative the negative eigenvalues of a $k$-locally PSD matrix can be. Our main result in this direction is a simple bound that is in terms of the \textbf{elementary symmetric polynomials}
\begin{align}\label{eq:elem_sym_poly}
    S_k(\lambda_1,\lambda_2,\ldots,\lambda_n) \defeq \sum_{1 \leq i_1 < \cdots < i_k \leq n} \left(\prod_{j=1}^k \lambda_{i_j}\right).
\end{align}

\begin{theorem}\label{thm:sym_poly}
    If $X \in \mathcal{I}_{k,n}^\circ$ has eigenvalues $\lambda_1$, $\lambda_2$, $\ldots$, $\lambda_n$ then
    \begin{align}\label{eq:sym_poly_eig}
        S_j(\lambda_1,\lambda_2,\ldots,\lambda_n) \geq 0 \quad \text{for all} \quad 1 \leq j \leq k.
    \end{align}
    Conversely, if $k \in \{1,n-1,n\}$ then, for any scalars $\lambda_1$, $\lambda_2$, $\ldots$, $\lambda_n \in \mathbb{R}$ satisfying Inequality~\eqref{eq:sym_poly_eig}, there exists $X \in \mathcal{I}_{k,n}^\circ$ with eigenvalues $\lambda_1$, $\lambda_2$, $\ldots$, $\lambda_n$.
\end{theorem}

Before we can prove this theorem, we need the following lemma:

\begin{lemma} \label{circlemma} \cite[Theorem 2.1]{kushel2016circulants} Let $C$ be an $n\times n$ circulant matrix with the monic characteristic polynomial $p(z)$. Then every $(n-1)\times(n-1)$ principal submatrix of $C$ has characteristic polynomial $\frac{1}{n}p^{\prime}(z)$.
\end{lemma}

We note that, while the original statement of this result given in \cite{kushel2016circulants} is in terms of the $(n-1)\times(n-1)$ \emph{upper left} principal submatrix of $C$, the result and proof work for all $(n-1)\times(n-1)$ principal submatrices of $C$.  In fact, any two $(n-1)\times(n-1)$ principal submatrices of $C$ can be seen to be permutationally similar and so have the same characteristic polynomial.

\begin{proof}[Proof of Theorem~\ref{thm:sym_poly}]
    To see that Inequality~\eqref{eq:sym_poly_eig} holds when $\lambda_1$, $\lambda_2$, $\ldots$, $\lambda_n$ are eigenvalues of $X \in \mathcal{I}_{k,n}^\circ$, recall from \cite[Theorem~1.2.16]{HJ13} that $S_j(\lambda_1,\lambda_2,\ldots,\lambda_n)$ equals the sum of the $j \times j$ principal minors of $X$. Since $X$ is $k$-locally PSD, those principal minors are all non-negative whenever $1 \leq j \leq k$, so $S_j(\lambda_1,\lambda_2,\ldots,\lambda_n) \geq 0$ whenever $1 \leq j \leq k$ as well.
    
    For the ``conversely'' direction, we first consider the case when $k = 1$. If $\lambda_1$, $\lambda_2$, $\ldots$, $\lambda_n$ are such that $\lambda_1 + \lambda_2 + \cdots + \lambda_n \geq 0$ (i.e., Inequality~\eqref{eq:sym_poly_eig} holds), then the Schur--Horn theorem \cite[Theorem~4.3.48]{HJ13} tells us that there exists $X \in M_n^{\textup{H}}$ with eigenvalues $\lambda_1$, $\lambda_2$, $\ldots$, $\lambda_n$ and each diagonal entry equal to $(\lambda_1 + \lambda_2 + \cdots + \lambda_n)/n$. Since these diagonal entries are all non-negative, $X \in \mathcal{I}_{1,n}^\circ$ as desired.
    
    Next, we consider the $k = n$ case. It is well-known that $S_j(\lambda_1,\lambda_2,\ldots,\lambda_n) \geq 0$ for all $1 \leq j \leq n$ is equivalent to $\lambda_j \geq 0$ for all $1 \leq j \leq n$. Since we can construct a positive semidefinite matrix (i.e., a matrix $X \in M_n^{+} = \mathcal{I}_{n,n}^\circ$) having any set of non-negative eigenvalues, we are done with this case.
    
    All that remains is to prove the ``conversely'' statement when $k=n-1$. Let $\lambda_1$, $\lambda_2$, $\ldots$, $\lambda_n$ be real numbers such that $S_j(\lambda_1,\lambda_2,\ldots,\lambda_n) \geq 0$ for all $1 \leq j \leq n-1$. Let $X \in M_n$ be a circulant matrix with monic characteristic polynomial
    \[
        p(z) = \prod_{j=1}^{n}(z-\lambda_i) = z^n+\sum_{j=1}^n(-1)^jS_j(\lambda_1,\lambda_2,\ldots,\lambda_n)z^{n-j}.
    \]
    Since the roots of $p$ are all real, $X$ must be Hermitian. Furthermore,
    \[
        p^{\prime}(z) = nz^{n-1} + \sum_{j=1}^{n-1}(-1)^j(n-j)S_j(\lambda_1,\lambda_2,\ldots,\lambda_n)z^{n-j-1}
    \]
    must also have all of its roots real. In fact, since the coefficients of $p^{\prime}(z)$ are alternating, all of its roots must be non-negative. Lemma~\ref{circlemma} then tells us that every $(n-1)\times(n-1)$ principal submatrix of $X$ has all eigenvalues non-negative, so we conclude that $X\in \mathcal{I}_{n-1,n}^\circ$, completing the proof.
\end{proof}

It is worth noting that Theorem~\ref{thm:sym_poly} gives a complete characterization of the possible spectra of $k$-locally PSD matrices (for all values of $k$) when $n = 3$. Given $\lambda_1$, $\lambda_2$, $\lambda_3 \in \mathbb{R}$ satisfying Inequality~\eqref{eq:sym_poly_eig}, it is straightforward to construct $X \in \mathcal{I}_{k,3}^\circ$ with eigenvalues $\lambda_1$, $\lambda_2$, $\lambda_3$ when $k = 1$ or $k = 3$. To construct such an $X$ when $k = 2$ (i.e., to make the ``conversely'' direction of the proof above a bit more explicit when $n = 3$ and $k = 2$), note that in this case we have
\[
    \lambda_1 + \lambda_2 + \lambda_3 \geq 0 \quad \text{and} \quad \lambda_1\lambda_2 + \lambda_1\lambda_3 + \lambda_2\lambda_3 \geq 0.
\]
If $U \in M_3$ is the Fourier matrix (i.e., $u_{j,\ell} = \mathrm{exp}\big(2\pi i(j-1)(\ell-1)/3\big)$ for all $1\leq j,\ell \leq 3$) then $X := U\mathrm{diag}(\lambda_1,\lambda_2,\lambda_3)U^*$ has eigenvalues $\lambda_1$, $\lambda_2$, and $\lambda_3$, and the proof of Theorem~\ref{thm:sym_poly} shows that $X \in \mathcal{I}_{2,3}^\circ$. To verify this claim a bit more directly, note that $X$ is a Hermitian circulant matrix and thus must have the form
\[
   X = \begin{bmatrix}
       a & b & \overline{b} \\
      \overline{b} & a & b \\
      b & \overline{b} & a
  \end{bmatrix}
\]
for some $a \in \mathbb{R}$ and $b \in \mathbb{C}$. Since $\tr(X) = 3a = \lambda_1 + \lambda_2 + \lambda_3$, we conclude that $a = (\lambda_1 + \lambda_2 + \lambda_3)/3 \geq 0$. Furthermore, the Frobenius norm of $X$ satisfies $\|X\|_{\textup{F}}^2 = 3a^2 + 6|b|^2 = \lambda_1^2 + \lambda_2^2 + \lambda_3^2$. If we substitute $a = (\lambda_1 + \lambda_2 + \lambda_3)/3$ into this formula and rearrange, we learn that
\begin{align*}
    |b|^2 & = \frac{1}{9}\Big( \big(\lambda_1^2 + \lambda_2^2 + \lambda_3^2\big) - \big( \lambda_1\lambda_2 + \lambda_1\lambda_3 + \lambda_2\lambda_3 \big) \Big).
\end{align*}
Since $\lambda_1\lambda_2 + \lambda_1\lambda_3 + \lambda_2\lambda_3 \geq 0$, it follows that $|b|^2 \leq (\lambda_1 + \lambda_2 + \lambda_3)^2/9 = a^2$. This implies that every $2 \times 2$ principal submatrix of $X$ is positive semidefinite, so $X \in \mathcal{I}_{2,3}^\circ$, as claimed.

The smallest open case is thus $n = 4$, $k = 2$, where Theorem~\ref{thm:sym_poly} gives the necessary conditions that every $X \in \mathcal{I}_{2,4}^\circ$ has eigenvalues $\lambda_1$, $\lambda_2$, $\lambda_3$, $\lambda_4 \in \mathbb{R}$ satisfying
\[
    S_1(\lambda_1,\lambda_2,\lambda_3,\lambda_4) \geq 0 \quad \text{and} \quad S_2(\lambda_1,\lambda_2,\lambda_3,\lambda_4) \geq 0,
\]
but it is not clear whether or not these conditions are tight (i.e., whether or not, given $\lambda_1$, $\lambda_2$, $\lambda_3$, $\lambda_4$ satisfying this pair of inequalities, we can construct $X \in \mathcal{I}_{2,4}^\circ$ with these eigenvalues). We present a computational method in Section~\ref{sec:computational_method_factor_pos} that can be used to construct a matrix $X \in \mathcal{I}_{k,n}^\circ$ with a given spectrum. This method numerically suggests that Theorem~\ref{thm:sym_poly} might be tight in the $k = 2$, $n = 4$ case, but a proof remains elusive.

\subsection{A Computational Method for Construction}\label{sec:computational_method_factor_pos}

We now present a numerical method that attempts to construct a $k$-locally PSD matrix with a given spectrum $\bm{\lambda} = (\lambda_1, \lambda_2, \ldots, \lambda_n)$ satisfying Inequalities~\eqref{eq:sym_poly_eig}. This algorithm is adapted from the method of \cite[Section~3.2.4]{CG05} for numerically constructing a matrix with specified eigenvalues in a given affine space of matrices.

Our goal is to construct a matrix that is in both the set $\mathcal{I}_{k,n}^\circ$ of matrices that are $k$-locally PSD, and also the set $M(\bm{\lambda})$ of matrices with the given spectrum:
\[
    M(\bm{\lambda}) \defeq \big\{ U\mathrm{diag}(\bm{\lambda}) U^* : U \text{ is unitary} \big\}.
\]
To find such a matrix, we try to compute
\begin{align}\label{eq:construction_norm}
    \min_{X \in \mathcal{I}_{k,n}^\circ, Y \in M(\bm{\lambda})} \|X - Y\|_{\textup{F}}.
\end{align}
(Here we have used the Frobenius norm, but many other matrix norms like the operator or trace norm would work just as well for our purposes.)

Indeed, there exists a matrix in $M(\bm{\lambda}) \cap \mathcal{I}_{k,n}^\circ$ if and only if this quantity equals $0$. Computing the quantity in Equation~\eqref{eq:construction_norm} explicitly is difficult due to the fact that $M(\bm{\lambda})$ is not convex. However, we can approximate it by iteratively bouncing back and forth between $M(\bm{\lambda})$ and $\mathcal{I}_{k,n}^\circ$, decreasing the value of the quantity~\eqref{eq:construction_norm} at each step, as illustrated schematically in Figure~\ref{fig:comp_iterate}.
\begin{figure}[!htb]
\begin{center}
 \includegraphics[scale=0.3]{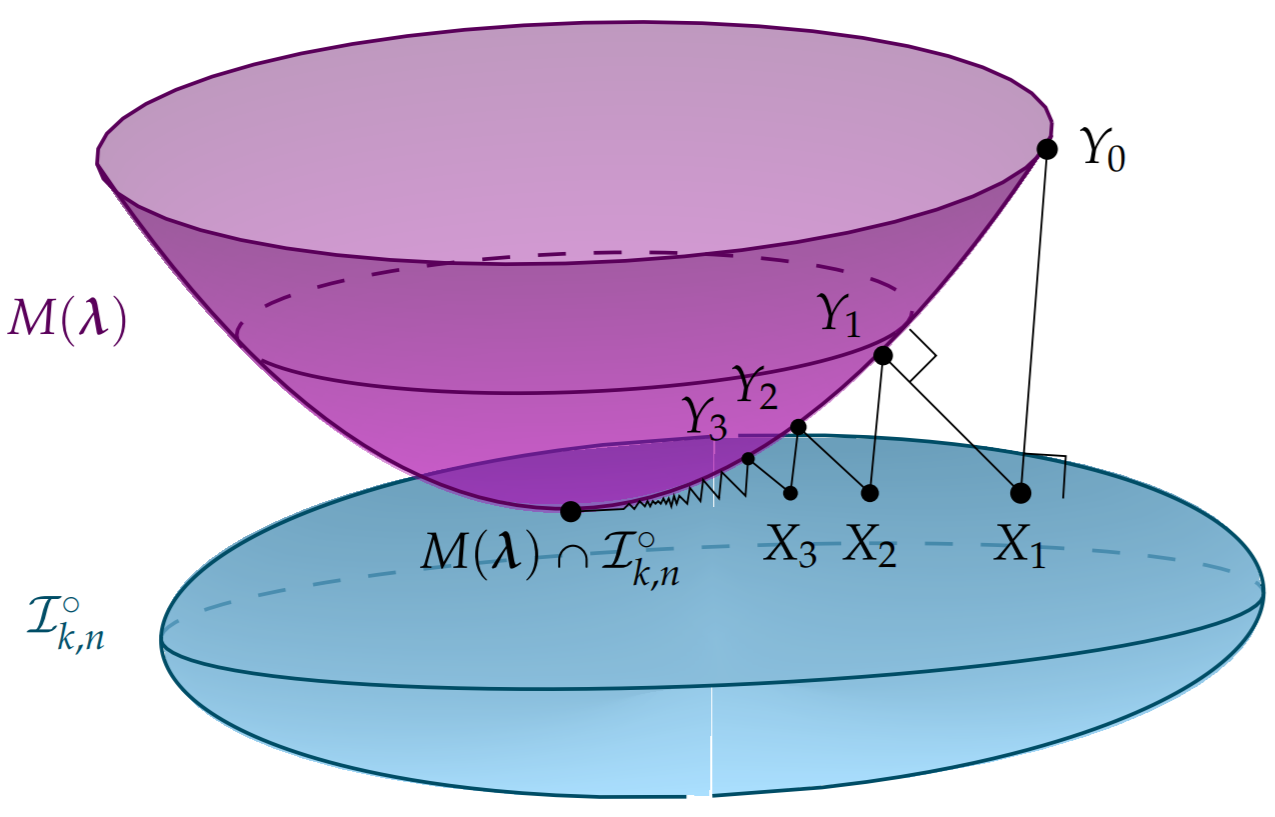}    
\end{center}\caption{A schematic of an iterative algorithm for finding a matrix in $M(\bm{\lambda}) \cap \mathcal{I}_{k,n}^\circ$ by bouncing back and forth between $M(\bm{\lambda})$ and $\mathcal{I}_{k,n}^\circ$.}
    \label{fig:comp_iterate}
\end{figure}

More explicitly, we let $Y_0 \in M(\bm{\lambda})$ be a random matrix with the desired spectrum (obtained by choosing $U$ to be a random unitary matrix and setting $Y_0 = U\mathrm{diag}(\bm{\lambda}) U^*$, for example). We then set $j = 1$ and proceed as follows:

\begin{itemize}
    \item Let
    \[
        X_j := \argmin_{X \in \mathcal{I}_{k,n}^\circ} \|X - Y_{j-1}\|_{\textup{F}}.
    \]
    We note that, since $\mathcal{I}_{k,n}^\circ$ is semidefinite-representable, $X_j$ can be computed via a semidefinite program (we direct the interested reader to any number of standard references, like \cite{Wat18}, for an introduction to semidefinite programming). In particular, $X_j$ can be computed via the following semidefinite program:
	\begin{align*}
    	\text{minimize: } & \ \tr(Z) \\
    	\text{subject to: } & \ \begin{bmatrix}
    	    I_n & X - Y_{j-1} \\ X - Y_{j-1} & Z
    	\end{bmatrix} \succeq O \\
    	& \ [X]_{J,J} \succeq O \text{ for each } J \subseteq [n] \text{ with } |J| = k,
	\end{align*}
	where $[n] = \{1,2,\ldots,n\}$ and $[X]_{J,J}$ is the $|J| \times |J|$ principal submatrix of $X$ that is obtained by taking the rows and columns of $X$ indexed by the members of $J$. The constraint that $[X]_{J,J} \succeq O$ for each $J \subseteq [n]$ with $|J| = k$ is simply equivalent to $X \in \mathcal{I}_{k,n}^\circ$. Furthermore, standard facts about Schur complements imply that, in this semidefinite program, we have $\tr(Z) = \|X - Y_{j-1}\|_{\textup{F}}^2$.
	
	\item Let
    \[
        Y_j := \argmin_{Y \in M(\bm{\lambda})} \|X_j - Y\|_{\textup{F}}.
    \]
    It follows from the equality condition of the Hoffman--Wielandt inequality \cite{hoffman1953} (which gives a lower bound for the Frobenius norm of the difference of two normal matrices) that if $X_j$ has spectral decomposition $X_j = UDU^*$, then $Y_j = U\Lambda U^*$, where $\Lambda$ is the diagonal matrix with diagonal entries $\lambda_1$, $\lambda_2$, $\ldots$, $\lambda_n$ arranged in the same relative order as the diagonal entries of $D$ (e.g., if the diagonal entries of $D$ are arranged largest-to-smallest, then so should be the diagonal entries of $\Lambda$).
    
    \item Increment $j$ by $1$ and return to the first bullet point, or stop once the difference between $\|X_j - Y_j\|_{\textup{F}}$ and $\|X_{j-1} - Y_{j-1}\|_{\textup{F}}$ falls below some threshold (i.e., once the algorithm has converged).
\end{itemize}

The above algorithm always terminates, since the sequence of values $\|X_j - Y_j\|_{\textup{F}}$ monotonically decreases:
\[
    \|X_j - Y_j\|_{\textup{F}} \leq \|X_j - Y_{j-1}\|_{\textup{F}} \leq \|X_{j-1} - Y_{j-1}\|_{\textup{F}} \quad \text{for all} \quad j \in \mathbb{N}.
\]
It is not clear whether or not the sequence of values $\|X_j - Y_j\|_{\textup{F}}$ always decreases \emph{to zero} (thus giving a matrix in the intersection $M(\bm{\lambda}) \cap \mathcal{I}_{k,n}^\circ$) when $\bm{\lambda}$ satisfies Inequalities~\eqref{eq:sym_poly_eig}, but we have implemented this method in MATLAB \cite{SuppCode} via the CVX package \cite{CVX}, and numerical tests suggest that it works quite well in practice. We illustrate with an example.

\begin{example}
    Consider the problem of constructing a $4 \times 4$ matrix $X$ that is $2$-locally PSD with spectrum $\bm{\lambda} = (10,4,-1,-2)$. Since $S_1(10,4,-1,-2) = 11 \geq 0$ and $S_2(10,4,-1,-2) = 0$, Inequalities~\eqref{eq:sym_poly_eig} hold so we are hopeful that this might be possible.
    
    After running the algorithm described above for 1,000 iterations, we find the following matrix:
    \[
        X \approx \begin{bmatrix}
            1.1278         &    0.0384 + 1.5436i     &    -0.9027 + 1.9431i     &     1.9983 + 0.4081i \\
            0.0384 - 1.5436i  &     2.1138           &    -0.7119 + 2.8455i     &    -1.1455 + 2.5464i \\
            -0.9027 - 1.9431i & -0.7119 - 2.8455i     &          4.0701         &      2.9167 - 2.5504i \\
            1.9983 - 0.4081i   &    -1.1455 - 2.5464i &   2.9167 + 2.5504i  & 3.6882
        \end{bmatrix}.
    \]
    It is straightforward to check that (within numerical precision) the eigenvalues of $X$ are indeed $10$, $4$, $-1$, and $-2$, and every $2 \times 2$ principal submatrix of $X$ is positive semidefinite (so $X$ is $2$-locally PSD).
\end{example}

\subsection{Hyperbolic Cones}\label{sec:hyperbolic_cones}

One way of rephrasing Theorem~\ref{thm:sym_poly} is as saying that the spectrum of every $k$-locally PSD matrix $X \in \mathcal{I}_{k,n}^\circ$ is a member of the cone
\begin{align}\label{eq:Ck_cone_defn}
    C_{k} \defeq \big\{ \bm{\lambda} \in \mathbb{R}^n : S_j(\lambda_1,\lambda_2,\ldots,\lambda_n) \geq 0 \ \text{for all} \ 1 \leq j \leq k \big\},
\end{align}
where $S_j$ is the $j$-th elementary symmetric polynomial from Equation~\eqref{eq:elem_sym_poly}. This cone $C_k$ is invariant under permutation of the entries of $\bm{\lambda}$, and it is furthermore closed and convex. However, convexity is not obvious, and follows from the fact that it is a \textbf{hyperbolic cone}---an object that has developed significant interest in its own right recently \cite{Bra12,R06,SP15}. We now provide a brief introduction to hyperbolicity cones, and we show that $C_k$ really is a hyperbolicity cone, and is thus convex.

\begin{definition}
    A homogeneous polynomial $P$ of degree $m\in \mathbb{N}$ (i.e. $P(\alpha \mathbf{x})=\alpha^m P(\mathbf{x})$ for all $\alpha\in \mathbb{R}$ and $\mathbf{x}\in \mathbb{R}^n$) called \textbf{hyperbolic} with respect to direction $\mathbf{d}\in \mathbb{R}^n$ if the univariate polynomial $p(t) := P(\mathbf{x}+t\mathbf{d})$ has all roots real for every $\mathbf{x}\in \mathbb{R}^n$.
\end{definition}

We can expand $P(\mathbf{x}+t\mathbf{d})$ as 
\begin{equation}\label{P_hyper_expand}
    P(\mathbf{x}+t\mathbf{d})=P(\mathbf{d})\big(t^m+P_1(\mathbf{x})t^{m-1}+P_2(\mathbf{x})t^{m-2}+\cdots+P_{m-1}(\mathbf{x})t+P_m(\mathbf{x})\big),
\end{equation}
for some polynomials $P_i(\mathbf{x})$ ($1 \leq i \leq m$) that are homogeneous of degree $i$. Then the \textbf{hyperbolicity cone} of $P$ with respect to the direction $\mathbf{d}$, denoted by $\Lambda(P,\mathbf{d})$, is the set \cite[Theorem~20]{R06}
\begin{equation}\label{Renegar_desc}
    \Lambda(P,\mathbf{d}) \defeq \{ \mathbf{x}\in\mathbb{R}^n : P_i(\mathbf{x})\geq 0\quad \text{for all} \quad 1 \leq i \leq m\}.
\end{equation}

If we choose $P(\mathbf{x})=S_n(\mathbf{x})=x_1x_2\cdots x_n$ then $P$ is homogeneous of degree $n$ with respect to the direction $\mathbf{d}=\bm{1} = (1,1,\ldots,1)$, since the univariate polynomial $p(t) = S_n(\mathbf{x}+t\bm{1})=\prod_{i=1}^n(x_i+t)$ has all roots real. Expanding $S_n$ as in Equation~\eqref{P_hyper_expand} then gives
\begin{equation}\label{P(S_n)}
    S_n(\mathbf{x}+t\bm{1})=t^n+S_1(\mathbf{x})t^{n-1}+\cdots+S_{n-1}(\mathbf{x})t+S_n(\mathbf{x}),
\end{equation}
where $S_j(\mathbf{x})$ is the elementry symmetric polynomial from Equation~\eqref{eq:elem_sym_poly}. Equation~\eqref{Renegar_desc} then tells us that the hyperbolicity cone of $S_n$ in the direction of $\bm{1}$ is
\[
    \Lambda(S_n,\bm{1}) = \{ \mathbf{x}\in\mathbb{R}^n : S_k(\mathbf{x})\geq 0\quad \text{for all} \quad 1 \leq k \leq n \},
\]
which is the cone $C_n$ from Equation~\eqref{eq:Ck_cone_defn}. We thus conclude that $C_n$ is a hyperbolic cone.

We can similarly see that $C_k$ is a hyperbolicity cone for all $1 \leq k \leq n-1$ by using directional derivatives, as follows. If $P$ is hyperbolic with respect to the direction $\mathbf{d}$, then the directional derivative of $P$ in the direction of $\mathbf{d}$,
\[
    p^\prime(0) = \frac{d}{dt}P(\mathbf{x} + t\mathbf{d})\Big|_{t=0},
\]
is also hyperbolic with respect to $\mathbf{d}$ \cite{Garding}. If we let $P^{(k)}$ denote the $k$-th directional derivative of $S_n$ in the direction of $\bm{1}$, then straightforward computation from Equation~\eqref{P(S_n)} shows that
\[
    P^{(k)}(\mathbf{x} + t\bm{1})=c_{n-k}t^{n-k}+c_{n-k-1}S_1(\mathbf{x})t^{n-k-1}+\cdots+c_1S_{n-k-1}(\mathbf{x})t+c_0S_{n-k}(\mathbf{x}),
\]
where $c_i=(i+k)! / i!$. Since $c_i > 0$, it then follows from Equation~\eqref{Renegar_desc} that the hyperbolicity cone of $P^{(k)}$ with respect to the direction $\bm{1}$ is
\[
    \Lambda(P^{(k)},\bm{1}) = \{ \mathbf{x}\in\mathbb{R}^n : S_j(\mathbf{x})\geq 0 \quad \text{for all} \quad 1 \leq j \leq n-k\},
\]
which is exactly the cone $C_k$ from Equation~\eqref{eq:Ck_cone_defn}. We thus conclude that $C_k$ is a hyperbolicity cone for all $1 \leq k \leq n$, and is thus convex by \cite[Theorem 2]{Garding}.

\section{Absolute k-Incoherence}\label{sec:abs_kcoh}

We now explore the central question of our work: Which quantum states can be determined to be $k$-incoherent based only on their spectrum (or equivalently, which matrices can be shown to have factor width at most $k$ based only on their spectrum)?

\begin{definition}\label{defn:abs_k_incoh}
    A mixed quantum state $\rho \in M_n^{+}$ is \textbf{absolutely $\mathbf{k}$-incoherent} if $U\rho U^* \in \mathcal{I}_{k,n}$ for all unitary matrices $U \in M_n$.
\end{definition}

We note that absolute $k$-incoherence is trivial if $k = 1$, since the only absolutely $1$-incoherent state is the maximally mixed state $\rho = I/n$ (a fact that follows immediately from the spectral decomposition and the fact that $\mathcal{I}_{1,n}$ is the set of diagonal density matrices). The other extreme is also trivial: if $k = n$ then $\mathcal{I}_{n,n}$ is the set of \emph{all} density matrices, so every mixed quantum state is absolutely $n$-incoherent.

However, if $2 \leq k \leq n-1$ then the set of absolutely $k$-incoherent states is non-trivial, and our goal is to characterize this set in a way that makes membership easy to check. Since absolute $k$-incoherence only depends on the spectrum of a state, our goal is to find inequalities or other simple-to-check properties of the spectrum that guarantee that the state is or is not absolutely $k$-incoherent.

In order to derive bounds on the set of absolutely $k$-incoherent quantum states, we will use duality of the convex cones $C_k$ defined in Equation~\eqref{eq:Ck_cone_defn}. In particular, the following result is our starting point:

\begin{theorem}\label{thm:dual_ck_abs_coh}
    Suppose $\rho \in M_n^{+}$ and let $\bm{\lambda} = (\lambda_1, \lambda_2, \ldots, \lambda_n)$ be a vector whose entries are eigenvalues of $\rho$ (in some order). If $\bm{\lambda} \in C_k^\circ$ then $\rho$ is absolutely $k$-incoherent. Furthermore, the converse holds if $k \in \{1,n-1,n\}$.
\end{theorem}

\begin{proof}
    Suppose $X \in \mathcal{I}_{k,n}^\circ$ and let $\bm{\mu} = (\mu_1, \mu_2, \ldots, \mu_n)$ be a vector whose entries are the eigenvalues of $X$. Since $\bm{\lambda} \in C_k^\circ$, and we know from Theorem~\ref{thm:sym_poly} that $\bm{\mu} \in C_k$, we conclude that $\bm{\lambda} \cdot \bm{\mu} \geq 0$. Since $C_k$ is invariant under permutations of its entries, this is equivalent to
    \begin{align}\label{ineq:lam_mu}
        \sum_{j=1}^{n} \lambda_j \mu_{\pi(j)} \geq 0
    \end{align}
    for all permutations $\pi : \{1,2,\ldots, n\} \rightarrow \{1,2,\ldots, n\}$.
    
    It then follows from \cite[Problem~III.6.14]{Bha97} that this is equivalent to $\tr(X U\rho U^*) \geq 0$ for all unitary matrices $U \in M_n$. Since $X \in \mathcal{I}_{k,n}^\circ$ was arbitrary, it follows that $U\rho U^* \in \mathcal{I}_{k,n}$ for all unitary matrices $U \in M_n$, which means exactly that $\rho$ is absolutely $k$-incoherent.
    
    For the ``furthermore'' statement, we just note that every step in this proof can be reversed (for all values of $k$), except for the appeal to Theorem~\ref{thm:sym_poly}. However if $k \in \{1,n-1,n\}$ then Theorem~\ref{thm:sym_poly} does tell us that $\bm{\mu} \in C_k$ implies $\bm{\mu}$ is a vector whose entries are the eigenvalues of some $X \in \mathcal{I}_{k,n}^\circ$, which completes the proof.
\end{proof}

Theorem~\ref{thm:dual_ck_abs_coh} raises a natural question---given a vector of eigenvalues $\bm{\lambda}$, how can we determine whether or not it is a member of $C_k^\circ$ (and thus whether or not the corresponding quantum state is absolutely $k$-incoherent)? These cones are actively being researched in the hyperbolic cone literature \cite{Bra12}, and recently it was shown that they admit a polynomial-size semidefinite representation \cite{SP15} (however, the specifics of that semidefinite representation are rather complicated, so we do not repeat them here). In particular, this means that membership in $C_k^\circ$ can be determined in polynomial time, so the hypotheses of Theorem \ref{thm:dual_ck_abs_coh} can be checked efficiently.

Our next result shows that absolutely $k$-incoherent states must have large rank (this is analogous to the already-known fact that absolutely separable states must have large rank \cite[Proposition~7.3]{JLN15}).

\begin{theorem}\label{thm:k_incoh_projection}
    Suppose $\rho \in M_n^{+}$ is a quantum state. If $\rank(\rho) \leq n-k$ then $\rho$ is not absolutely $k$-incoherent. Furthermore, this bound is tight; the rank-$(n-k+1)$ state with its non-zero eigenvalues equal to each other is absolutely $k$-incoherent.
\end{theorem}

We need the following lemma to help us prove the above theorem. This lemma's proof is from \cite{SpeyMO}, but we include it here for completeness.

\begin{lemma}\label{Speyer}
    Let $1 \leq k \leq n$ be integers. If $S_j(\lambda_1, \lambda_2, \cdots, \lambda_n) \geq 0$ for all $1\leq j\leq k$ then
    \[
        S_j(\lambda_1, \lambda_2, \cdots, \lambda_{n-1})\geq 0 \quad \text{for all} \quad 1\leq j\leq k-1.
    \]
\end{lemma}

\begin{proof}
    For brevity, we define $a_j := S_j(\lambda_1, \lambda_2, \cdots, \lambda_n)$ and $b_j := S_j(\lambda_1, \lambda_2, \cdots, \lambda_{n-1})$. Our goal is to show that if $a_j\geq 0$ for all $1\leq j\leq k$ then $b_j\geq 0$ for all $1\leq j\leq k-1$. We will prove the contrapositive: we will show that if there exists $1\leq j\leq k-1$ with $b_j < 0$ then there exists $1\leq j^\prime\leq k$ with $a_{j^\prime} < 0$ too.
    
    Without loss of generality, we may assume that $j$ is minimal (i.e., $b_j < 0$, but $b_1, b_2, \cdots, b_{j-1}\geq 0$). Consider the following two cases:
    \begin{itemize}
        \item[Case 1:] $\lambda_n<0$. It is straightforward to show that $a_j = \lambda_n b_{j-1} + b_j$, which must be strictly negative since $\lambda_n < 0$, $b_j < 0$, and $b_{j-1} \geq 0$. That is, we can choose $j^\prime = j$.
        
        \item[Case 2:] $\lambda_n\geq 0$. As in Case~1, we know that $a_j = \lambda_n b_{j-1} + b_j$. It follows that if $b_{j-1} \leq 0$ then $a_j<0$. A similar argument shows that if $b_{j+1} \leq 0$ then $a_{j+1} < 0$. We thus assume for the remainder of this Case that $b_{j-1}, b_{j+1} > 0$. By using Newton's inequalities, we see that 
        \[
            \frac{b_{j-1} b_{j+1}}{\binom{n-1}{j-1} \binom{n-1}{j+1}} \leq \frac{b_j^2}{\binom{n-1}{j}^2},
        \]
        which simplifies as
        \[
            b_{j-1} b_{j+1} \leq \frac{j(n-j-1)}{(j+1)(n-j)} b_j^2 < b_j^2.
        \]
        Keeping in mind that $b_j < 0$, this tells us that
        \[
            0 \leq \frac{b_{j+1}}{- b_j} < \frac{-b_j}{b_{j-1}}.
        \]
        We must either have $\lambda_n > \tfrac{b_{j+1}}{-b_j}$ or $\lambda_n < \tfrac{-b_j}{b_{j-1}}$ (or both). If $\lambda_n > \tfrac{b_{j+1}}{-b_j}$ then $a_{j+1} = \lambda_n b_{j} + b_{j+1} < 0$ (so we can choose $j^\prime = j+1$), and if $\lambda_n < \tfrac{-b_j}{b_{j-1}}$ then $a_j = \lambda_n b_{j-1} + b_j < 0$ (so we can choose $j^\prime = j$).
    \end{itemize}
    With both cases taken care of, the proof is complete.
\end{proof}

\begin{proof}[Proof of Theorem~\ref{thm:k_incoh_projection}]
    To see that there does not exist an absolutely $k$-incoherent state $\rho$ with $\rank(\rho) \leq n-k$, recall from Theorem~\ref{thm:neg_eigs_of_kfacpos} that there exists $X \in \mathcal{I}_{k,n}^\circ$ with exactly $n-k$ eigenvalues. If $\rank(\rho) \leq n-k$ then there exists a unitary matrix $U \in M_n$ so that $U\rho U^*$ is supported on that $(n-k)$-dimensional negative eigenspace, so $\tr\big(X(U\rho U^*)\big) < 0$, so $U\rho U^*$ is not $k$-incoherent, so $\rho$ is not absolutely $k$-incoherent.
    
    Now let $\rho$ be the rank-$(n-k+1)$ state with its non-zero eigenvalues equal to each other. By using Lemma~\ref{Speyer} a total of $k-1$ times, we see that $(\lambda_1,\ldots,\lambda_n) \in C_k$ implies $\lambda_1 + \lambda_2 + \cdots + \lambda_{n-k+1} \geq 0$. In particular, this means that if $X \in \mathcal{I}_{k,n}^\circ$ has eigenvalues $\lambda_1 \geq \lambda_2 \geq \cdots \geq \lambda_n$ then
    \[
        \tr(X\rho) \geq \frac{\lambda_1 + \lambda_2 + \cdots + \lambda_{n-k+1}}{n-k+1} \geq 0,
    \]
    so $\rho$ is absolutely $k$-incoherent.
\end{proof}

Theorem~\ref{thm:k_incoh_projection} provides us with a single non-trivial example of a mixed state that is absolutely $k$-incoherent. Our next theorem provides many more:

\begin{theorem}\label{thm:k_incoh_max_eig}
    Suppose $\rho \in M_n^{+}$ has maximal eigenvalue $\lambda_{\text{max}}$. If $\lambda_{\text{max}} \leq 1/(n-k+1)$ then $\rho$ is absolutely $k$-incoherent.
\end{theorem}

Once again, we need a lemma to help us prove the above theorem. In fact, we also need to introduce some additional notation and terminology. Recall that if $\mathbf{x}, \mathbf{y} \in \mathbb{R}^n$ then we say that $\mathbf{x}$ is \textbf{majorized} by $\mathbf{y}$, denoted $\mathbf{x} \prec \mathbf{y}$, if 
\[
    \sum_{j=1}^k x_j^{\downarrow} \leq \sum_{j=1}^k y_j^{\downarrow} \quad \text{for all} \quad 1 \leq k \leq n,
\]
with equality when $k=n$, where $x_j^\downarrow$ refers to the $j$-th largest entry of $\mathbf{x}$ (i.e., $x_1^\downarrow \geq x_2^\downarrow \geq \cdots \geq x_n^\downarrow$). Given $X,Y \in M_n^{\textup{H}}$, we say that $X$ is majorized by $Y$, and we write $X \prec Y$, if the vector of eigenvalues of $X$ is majorized by the vector of eigenvalues of $Y$.

\begin{lemma}\label{lem:k_incoh_majorize}
   Suppose $\rho,\sigma \in M_n^{+}$ are such that $\rho \prec \sigma$. If $\sigma$ is absolutely $k$-incoherent, then so is $\rho$.
\end{lemma}

\begin{proof}
    Since $\rho\prec \sigma$, we can use Uhlmann's Theorem \cite[Equation~(3)]{nielsen2001majorization} to see that there exists a probability vector $(p_1,\cdots, p_r)$ and a finite collection of unitary matrices $\{U_j\} \subset M_n$ such that 
    \[
        \rho = \sum_{j=1}^r p_jU_j\sigma U_j^*.
    \]
    Since $\sigma$ is absolutely $k$-incoherent, so is $U_j\sigma U_j^*$ for each $1 \leq j \leq r$. Since the set of absolutely $k$-incoherent states is convex, $\rho$ must be absolutely $k$-incoherent too.
\end{proof}

It is worth noting that the argument used in the proof of the above lemma actually works in the ``absolute'' version of \emph{any} convex quantum resource theory, not just that of $k$-coherence \cite{CG18}: the only properties of the set of absolutely $k$-incoherent states that we used were the facts that it is convex and invariant under unitary conjugation. In fact, Lemma~\ref{lem:k_incoh_majorize} was noted in the context of ``absolute'' versions of resource theories of bipartite sets of states in \cite[Lemma~2.2]{JLN15}, via essentially the same proof.

\begin{proof}[Proof of Theorem~\ref{thm:k_incoh_max_eig}]
    Let $\rho$ be a mixed state with maximal eigenvalue satisfying $\lambda_{\text{max}} \leq 1/(n-k+1)$, and let $\sigma$ be a mixed state of rank $n-k+1$ with all of its non-zero eigenvalues equal to $1/(n-k+1)$. We know from Theorem~\ref{thm:k_incoh_projection} that $\sigma$ is absolutely $k$-incoherent. Since $\rho$ is majorized by $\sigma$, it follows from Lemma~\ref{lem:k_incoh_majorize} that $\rho$ is absolutely $k$-incoherent.
\end{proof}

We now look at how we can simplify Theorem~\ref{thm:dual_ck_abs_coh}, and make it more explicit, for certain specific values of $k$. In particular, we will focus on how we can characterize the dual cones $C_k^\circ$ in a way that makes membership in them easy to check in the $k \leq 2$ and $k \geq n-1$ cases.

It is straightforward to show that the dual cones of $C_1$ and $C_n$ satisfy $C_1^{\circ} = \{ c\bm{1} : c \geq 0 \} \subseteq \mathbb{R}^n$, where $\bm{1} = (1,1,\ldots,1)$, and $C_n^{\circ} = C_n \subseteq \mathbb{R}^n$ is the set of vectors with non-negative entries, which recovers our earlier observations about the sets of absolutely $1$-incoherent and absolutely $n$-incoherent states being trivial. The next two subsections are devoted to the less straightforward $k = 2$ and $k = n-1$ cases.

\subsection{Absolute 2-Incoherence}\label{sec:abs_2_incoh}

The following theorem completely characterizes the dual cone of $C_k$ in the $k = 2$ case. The norm $\|\bm{\lambda}\|$ that is used in it is the standard Euclidean norm on $\mathbb{R}^n$: $\|\bm{\lambda}\| := \sqrt{\sum_{j=1}^n \lambda_j^2}$.

\begin{theorem}\label{thm:2_incoh_dual}
    Let $C_2$ be as defined in Equation~\eqref{eq:Ck_cone_defn}. Then
    \[
        C_2^{\circ} = \left\{ \bm{\lambda} \in \mathbb{R}^n : \sum_{j=1}^n \lambda_j \geq \sqrt{n-1}\|\bm{\lambda}\| \quad \text{and} \quad \lambda_j \geq 0 \ \text{for all} \ 1 \leq j \leq n \right\}.
    \]
\end{theorem}

\begin{proof}
    We can rewrite $C_2$ in the following form:
    \[
        C_2 = \big\{ \bm{\lambda} \in \mathbb{R}^n : \|\bm{\lambda}\| \leq \bm{\lambda} \cdot \bm{1} \big\} = \left\{ \bm{\lambda} \in \mathbb{R}^n : \frac{1}{\sqrt{n}}\|\bm{\lambda}\| \leq \bm{\lambda} \cdot \big(\bm{1} / \sqrt{n}\big) \right\},
    \]
    where $\bm{1} = (1,1,\ldots,1) \in \mathbb{R}^n$. It follows that $C_2$ is a circular cone \cite{ZC13} with angle $\theta$ satisfying $\cos(\theta) = 1/\sqrt{n}$ (the referenced paper deals primarily with circular cones centered around $\mathbf{e_1} := (1,0,0,\ldots,0)$, whereas ours is rotated to be centered around the unit vector $\bm{1} / \sqrt{n}$). It follows from \cite[Theorem 2.1(c)]{ZC13} that the dual of $C_2$ is the circular cone centered around $\bm{1} / \sqrt{n}$ with angle $\pi/2 - \theta$. That is,
    \[
        C_2^\circ = \Big\{ \bm{\lambda} \in \mathbb{R}^n : \cos\big(\pi/2 - \arccos(1/\sqrt{n})\big)\|\bm{\lambda}\|_2 \leq \bm{\lambda} \cdot \big(\bm{1} / \sqrt{n}\big) \Big\}.
    \]
    Using the fact that $\cos\big(\pi/2 - \arccos(1/\sqrt{n})\big) = \sqrt{1 - 1/n}$, and then rearranging and simplifying, shows that $C_2^\circ$ has the form described in the statement of the theorem.
\end{proof}

By combining Theorem~\ref{thm:2_incoh_dual} with our earlier results, we immediately get simple-to-check sufficient condition for absolute $2$-incoherence. Furthermore, this condition is both necessary and sufficient in small dimensions:

\begin{theorem}\label{thm:abs_2incoh}
    Suppose $\rho \in M_n^{+}$ is a quantum state with eigenvalues $\lambda_1$, $\lambda_2$, $\ldots$, $\lambda_n$. If
    \[
        \sum_{j=1}^n \lambda_j^2 \leq \frac{1}{n-1}
    \]
    then $\rho$ is absolutely $2$-incoherent. The converse holds if $n \leq 3$.
\end{theorem}

\begin{proof}
    If $(\lambda_1, \lambda_2, \ldots, \lambda_n) \in C_2^\circ$ then Theorem~\ref{thm:dual_ck_abs_coh} tells us that $\rho$ is absolutely $2$-incoherent. By Theorem~\ref{thm:2_incoh_dual}, membership in $C_2^\circ$ is equivalent to
    \[
        \left(\sum_{j=1}^n \lambda_j\right)^2 \geq (n-1)\sum_{j=1}^n \lambda_j^2,
    \]
    which (when we use the fact that $\sum_{j=1}^n \lambda_j = 1$ since $\tr(\rho) = 1$) is equivalent to the inequality given in the statement of this theorem.
    
    To see that the converse holds when $n = 2$, recall that every mixed state is absolutely $2$-incoherent in this case, and for every mixed state we have
    \[
        \sum_{j=1}^n \lambda_j^2 = \|\rho\|_{\textup{F}}^2 \leq \|\rho\|_{\textup{tr}}^2 = 1 = \frac{1}{n-1},
    \]
    where $\|\rho\|_{\textup{F}}$ and $\|\rho\|_{\textup{tr}}$ are the Frobenius and trace norms of $\rho$, respectively.
    
    To see that the converse holds when $n = 3$, recall from Theorem~\ref{thm:sym_poly} that every $\bm{\mu} = (\mu_1,\mu_2,\mu_3) \in C_2$ occurs as a vector of eigenvalues of some matrix $X \in \mathcal{I}_{2,3}^\circ$. Working through the proof of Theorem~\ref{thm:dual_ck_abs_coh} backwards then shows that if $\rho$ is absolutely $2$-incoherent then Inequality~\eqref{ineq:lam_mu} holds for \emph{all} $\bm{\mu} \in C_2$, so $\bm{\lambda} \in C_2^\circ$.
\end{proof}

It is worth noting that the inequality in the above theorem is equivalent to $\|\rho\|_{\textup{F}}^2 \leq 1/(n-1)$. Alternatively, this is equivalent (again, via the fact that $\sum_{j=1}^n \lambda_j = 1$) to
\[
    \left\|\left(\frac{n-1}{n}\right)\rho - \frac{1}{n}I\right\|_{\textup{F}} \leq \frac{1}{n}.
\]
In particular, if $n = 3$ then we see that $\rho$ is absolutely $2$-incoherent if and only if $\big\|(2/3)\rho - I/3\big\|_{\textup{F}} \leq 1/3$, so the set of absolutely $2$-incoherent states is simply a ball (in the Frobenius norm) centered at the maximally mixed state $I/3$. This is in stark contrast with the absolute \emph{separability} problem, where there are indeed Frobenius balls of absolutely separable states \cite{GB02}, but there are also absolutely separable states in no such ball even in the smallest non-trivial dimensions \cite{KZ00}.

\subsection{Absolute $(n-1)$-Incoherence}\label{sec:abs_n1_incoh}

When $k = n-1$, Theorem~\ref{thm:dual_ck_abs_coh} says that $\rho$ is absolutely $k$-incoherent if and only if its vector of eigenvalues $\bm{\lambda}$ satisfies $\bm{\lambda} \in C_k^\circ$. The following theorem illustrates how to efficiently test whether or not $\bm{\lambda} \in C_k^\circ$, thus giving us an effective test for absolute $(n-1)$-incoherence.

\begin{theorem}\label{thm:abs_n_min1_incoherence}
    Suppose $\rho \in M_n^+$ is a quantum state with eigenvalues $\lambda_1 \geq \lambda_2 \geq \cdots \geq \lambda_n$. Then $\rho$ is absolutely $(n-1)$-incoherent if and only if there exists a positive semidefinite matrix $\Lambda \in M_n^{+}(\mathbb{R})$ such that
    \begin{align*}
        \lambda_1 & = -\Lambda_{1,1} - \sum_{i=2}^{n} \left(\Lambda_{1,i} + \Lambda_{i,1}\right) \quad \text{and} \quad \lambda_j = \Lambda_{j,j} \ \ \ \text{for ${} \ 2 \leq j \leq n$}.
    \end{align*}
\end{theorem}

\begin{proof}
    By Theorem~\ref{thm:dual_ck_abs_coh}, we know that $\rho$ is absolutely $(n-1)$-incoherent if and only if $\bm{\lambda} = (\lambda_1,\lambda_2,\ldots,\lambda_n) \in C_{n-1}^\circ$. This cone was characterized in \cite[Proposition~4.2]{Z8}, which says that
    \[
        C_{n-1}^\circ = \bigcap_{i=1}^n {C^{i}}^{\circ},
    \]
    where 
    \[
        {C^{i}}^{\circ} = \left\{ \mathbf{y} \in \mathbb{R}^n: ~-y_i = \sum_{k\neq i} (\Lambda_{k,i}+\Lambda_{i,k})+\Lambda_{i,i}, ~~y_j=\Lambda_{j,j} ~~\text{for}~~ j\neq i,~~\Lambda \ \text{is positive semidefinite}\right\}.
    \]
    
    Since vectors in each of the ${C^{i}}^{\circ\circ} = C^i$ cones can have at most $1$ negative entry (the $i$-th entry, as was discussed in \cite[Section~4]{Z8}), are identical up to permutation of their entries, and we have sorted the entries of $\bm{\lambda}$ so that its first entry is its largest, we conclude that $\bm{\lambda} \in \bigcap_{i=1}^n {C^{i}}^{\circ}$ is equivalent to $\bm{\lambda} \in {C^{1}}^{\circ}$. This completes the proof.
\end{proof}

The characterization of absolute $(n-1)$-incoherence that is provided by Theorem~\ref{thm:abs_n_min1_incoherence} is straightforward to check numerically via semidefinite programming. We provide MATLAB code that implements this semidefinite program, and thus checks whether or not a given quantum state is absolutely $(n-1)$-incoherent, via the CVX package \cite{CVX} at \cite{SuppCode}.

Theorem~\ref{thm:abs_n_min1_incoherence} can also be used to derive more explicit tests for absolute $(n-1)$-incoherence, at the expense of no longer being both necessary and sufficient. For example, we have the following necessary condition, which is a stronger version of the observation that no pure states can possibly be absolutely $k$-incoherent when $k \leq n-1$:

\begin{corollary}\label{cor:abs_n1_necc}
    Suppose $\rho \in M_n^+$ is a quantum state maximal eigenvalue $\lambda_{\text{max}}$. If $\rho$ is absolutely $(n-1)$-incoherent then
    \begin{align*}
        \lambda_{\text{max}} \leq 1 - \frac{1}{n}.
    \end{align*}
\end{corollary}

\begin{proof}[Proof of Corollary~\ref{cor:abs_n1_necc}.]
    Let $\lambda_1 \geq \lambda_2 \geq \cdots \geq \lambda_n \geq 0$ denote the eigenvalues of $\rho$ (so that, in particular, $\lambda_1 = \lambda_{\text{max}}$). If $\rho$ is absolutely $(n-1)$-incoherent then a matrix $\Lambda$ as described by Theorem~\ref{thm:abs_n_min1_incoherence} exists. Since $\Lambda$ is positive semidefinite, so are its $2\times 2$ principal submatrices, so we must have
    \[
        |\Lambda_{1,j}|^2 \leq \Lambda_{1,1}\Lambda_{j,j} = \lambda_j\left(-\lambda_1 - \sum_{i=2}^{n} \left(\Lambda_{1,i} + \Lambda_{i,1}\right)\right)
    \]
    for all $2 \leq j \leq n$. If we move $\lambda_1$ to the left-hand side and then sum over $j$, this tells us that
    \[
        \lambda_1\sum_{j=2}^{n}\lambda_j \leq -\sum_{j=2}^n \left(|\Lambda_{1,j}|^2 + \lambda_j \sum_{i=2}^{n} \left(\Lambda_{1,i} + \Lambda_{i,1}\right)\right) = -\sum_{i=2}^{n} \left(|\Lambda_{1,i}|^2 + \left(\sum_{j=2}^{n}\lambda_j\right)(\Lambda_{1,i} + \Lambda_{i,1})\right).
    \]
    Standard calculus-based optimization techniques show that the quantity on the right (if regarded as a function in the $n-1$ variables $\Lambda_{1,i}$, for $2 \leq i \leq n$) is maximized exactly when $\Lambda_{1,2} = \cdots = \Lambda_{1,n} = -\sum_{j=2}^{n}\lambda_j$, and its maximal value is $(n-1)\left(\sum_{j=2}^{n}\lambda_j\right)^2$. Simplifying then shows that
    \begin{align}\label{eq:lam1_sum2_ineq}
        \lambda_1 \leq (n-1)\sum_{j=2}^{n}\lambda_j.
    \end{align}
    
    Finally, if we use the fact that $\tr(\rho) = \sum_{j=1}^n\lambda_j = 1$, then we see that $\sum_{j=2}^n\lambda_j = 1-\lambda_1$. Plugging this equation into Inequality~\ref{eq:lam1_sum2_ineq} gives us the inequality in the statement of the theorem.
\end{proof}

Unlike Theorem~\ref{thm:abs_n_min1_incoherence}, Corollary~\ref{cor:abs_n1_necc} really is just a necessary condition for absolute $(n-1)$-incoherence. This is illustrated by the state $\rho = \mathrm{diag}(2/3,1/3,0) \in M_3^{+}$, which satisfies the hypotheses of Corollary~\ref{cor:abs_n1_necc}, but is not absolutely $(n-1) = 2$-incoherent (as can be seen from Theorem~\ref{thm:abs_2incoh}). However, Corollary~\ref{cor:abs_n1_necc} is still tight, in the sense that it is the best possible inequality depending only on $\lambda_1$ and $\tr(\rho)$:

\begin{example}\label{exam:n1_incoh_ness}
    Suppose $X \in M_n^{+}$ has eigenvalues $c \geq 1$ with multiplicity $1$, and $1$ with multiplicity $n-1$. We claim that the quantum state $\rho = X/\tr(X)$ is absolutely $(n-1)$-incoherent if and only if $c \leq (n-1)^2$.
    
    To see that $c = (n-1)^2$ implies absolute $(n-1)$-incoherence, choose the matrix $\Lambda$ in Theorem~\ref{thm:abs_n_min1_incoherence} to be $\Lambda = DJD^*/\tr(X)$, where $J$ is the all-ones matrix (which is positive semidefinite, so $\Lambda$ is too) and $D = \mathrm{diag}(1-n,1,1,\ldots,1)$. Then $\lambda_j = \Lambda_{j,j} = 1/\tr(X)$ for $2 \leq j \leq n$ and
    \[
        \lambda_1 = -\Lambda_{1,1} - \sum_{i=2}^{n} \left(\Lambda_{1,i} + \Lambda_{i,1}\right) = \frac{1}{\tr(X)}\big(-(n-1)^2 + 2(n-1)^2\big) = \frac{c}{\tr(X)},
    \]
    so Theorem~\ref{thm:abs_n_min1_incoherence} tells us that $\rho$ is absolutely $(n-1)$-incoherent. The fact that $1 \leq c < (n-1)^2$ implies absolute $(n-1)$-incoherence then follows from convexity.
    
    Conversely, if $\rho$ is absolutely $(n-1)$-incoherent then Corollary~\ref{cor:abs_n1_necc} applies, so we must have $c/\tr(X) = \lambda_1 \leq (n-1)\sum_{j=2}^n\lambda_j/\tr(X) = (n-1)^2/\tr(X)$, so $c \leq (n-1)^2$.
\end{example}

We already saw an explicit \emph{sufficient} (rather than necessary) condition for absolute $(n-1)$-incoherence back in Theorem~\ref{thm:k_incoh_max_eig}: plugging $k = n-1$ into that theorem tells us that if the maximal eigenvalue $\lambda_{\text{max}}$ of a quantum state $\rho \in M_n^{+}$ satisfies
\[
    \lambda_{\text{max}} \leq \frac{1}{2}
\]
then $\rho$ must be absolutely $(n-1)$-incoherent. We summarize these observations in Figure~\ref{fig:max_eigenvalue_abs_n1}.

\begin{figure}[!htb]
    \begin{center}
    	\begin{tikzpicture}[xscale=14.5]
    		\draw[thick] (0,0) -- (1,0);
    				
    		\draw[thick] (0,0.3) -- (0,-0.3);
    		\draw[thick] (0.1666666,0.2) -- (0.1666666,-0.2);
    		\draw[thick] (0.5,0.2) -- (0.5,-0.2);
    		\draw[thick] (0.8333333,0.2) -- (0.8333333,-0.2);
    		\draw[thick] (1,0.3) -- (1,-0.3);
    		
    		\draw[thick] (-0.02,0) node[anchor=east] {$\lambda_{\text{max}}$:};
    		
    		\draw[thick] (0,-0.22) node[anchor=north] {$0$};
    		\draw[thick] (0.1666666,-0.22) node[anchor=north] {$\frac{1}{n}$};
    		\draw[thick] (0.5,-0.22) node[anchor=north] {$\frac{1}{2}$};
    		\draw[thick] (0.8333333,-0.22) node[anchor=north] {$1-\frac{1}{n}$};
    		\draw[thick] (1,-0.22) node[anchor=north] {$1$};
    		
    		\draw[thick] (0.083333,0.09) node[anchor=south] {\footnotesize impossible};
    		\draw[thick] (0.333333,0.09) node[anchor=south,align=center] {\footnotesize Theorem~\ref{thm:k_incoh_max_eig}:\\\footnotesize abs. $(n-1)$-incoherent};
    		\draw[thick] (0.666666,0.09) node[anchor=south,align=center] {\footnotesize Example~\ref{exam:n1_incoh_ness}:\\\footnotesize maybe abs. $(n-1)$-incoherent};
    		\draw[thick] (0.916666,0.09) node[anchor=south,align=center] {\footnotesize Corollary~\ref{cor:abs_n1_necc}:\\\footnotesize not abs. \\\footnotesize $(n-1)$-incoh.};
    		
            \draw[<-] (0.166666,0.35) to (0.166666,1.5) node[anchor=south]{\small maximally mixed state};
    	\end{tikzpicture}
	\end{center}
	\caption{How the maximal eigenvalue $\lambda_{\text{max}}$ of a quantum state $\rho \in M_n^{+}$ relates to its absolute $(n-1)$-incoherence.}\label{fig:max_eigenvalue_abs_n1}
\end{figure}
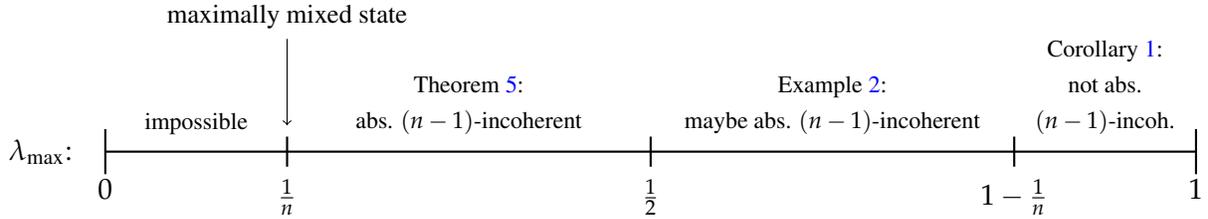

It is clear that any state $\rho \in M_n^+$ that is $k$-incoherent will still be $k$-incoherent if we regard it as a member of $M_{n+1}^+$ by adding rows and columns of zeroes. The following result shows that, perhaps surprisingly, an even stronger statement is true of \emph{absolute} $k$-incoherence, at least when $k = n-1$:

\begin{corollary}\label{cor:abs_n1_incoh_inc_size}
    A quantum state $\rho \in M_n^+$ is absolutely $(n-1)$-incoherent if and only if the state $\widetilde{\rho} \in M_{n+1}^+$, obtained by appending a row and column of zeros, is absolutely $n$-incoherent.
\end{corollary}

\begin{proof}
    This follows immediately from Theorem~\ref{thm:abs_n_min1_incoherence}: if $\rho$ is absolutely $(n-1)$-incoherent then there exists a positive semidefinite matrix $\Lambda \in M_n^+$ satisfying the constraints described by that theorem. Simply appending a row and column of zeros to $\Lambda$ to obtain a new matrix $\widetilde{\Lambda} \in M_{n+1}^+$ shows, again via Theorem~\ref{thm:abs_n_min1_incoherence}, that $\widetilde{\rho}$ is absolutely $n$-incoherent.
    
    Conversely, if $\widetilde{\rho} \in M_{n+1}^+$ is absolutely $n$-incoherent and has its last row and column consisting entirely of zeros, then the matrix $\widetilde{\Lambda} \in M_{n+1}^+$ described by Theorem~\ref{thm:abs_n_min1_incoherence} must have $\Lambda_{n+1,n+1} = 0$ and thus its last row and column must consist entirely of zeros. Erasing that final row and column then gives a matrix $\Lambda \in M_n^{+}$ that shows, again via Theorem~\ref{thm:abs_n_min1_incoherence}, that $\rho$ is absolutely $(n-1)$-incoherent.
\end{proof}

\section{Conclusions and Open Questions}\label{sec:conclusions}

In this work, we explored spectral problems of interest in the quantum resource theory of $k$-coherence. In particular, we introduced $k$-locally PSD matrices as witnesses of $k$-coherence, and we derived bounds on the spectra of $k$-locally PSD matrices. We also introduced the set of \emph{absolutely} $k$-incoherent quantum states, and used our spectral bounds for $k$-locally PSD matrices to produce bounds on the set of $k$-incoherent quantum states.

While we have solved numerous problems, our work leaves open many questions and avenues for future research:

\begin{itemize}
    \item The computational method of Section~\ref{sec:computational_method_factor_pos} for constructing $k$-locally PSD matrices with a given spectrum seems to work well in practice. Are there conditions under which it is \emph{guaranteed} to work (i.e., the iteration is guaranteed to converge to a $k$-locally PSD matrix with the desired spectrum)?
    
    \item We established numerous results about the spectrum of $k$-locally PSD matrices and absolute $k$-incoherence. Since there is a connection between $k$-incoherence and $k$-entanglement \cite{SS15}, it seems natural to ask whether or not any of our results can be extended in a useful way to the spectrum of $k$-entanglement witnesses and absolute $k$-entanglement.
    
    \item Theorem~\ref{thm:abs_2incoh} gives a sufficient condition for absolute $2$-incoherence that is also necessary in dimension $n = 3$. It is unclear whether or not it is necessary in dimension $n = 4$ or higher. Similarly, we expect that the ``conversely'' statements of Theorems~\ref{thm:sym_poly} and~\ref{thm:dual_ck_abs_coh} are false when $n$ and/or $k$ are large enough, but we do not have any explicit examples to demonstrate this.
    
    \item Theorem~\ref{thm:abs_n_min1_incoherence} shows that absolute $(n-1)$-incoherence can be determined in polynomial time. Is the same true of absolute $k$-incoherence for all $1 \leq k \leq n$? It is worth noting that the set of ``absolutely PPT states'' \cite{Hil07} (which we do not define here) has a similar semidefinite representation, but membership in that set is not known to be checkable in polynomial time since the number of matrix inequalities that need to be checked grows exponentially.
    
    \item Does Corollary~\ref{cor:abs_n1_incoh_inc_size} hold for $k$-incoherence when $k \neq n-1$? That is, is it true that absolute $k$-incoherence of $\rho \in M_n^+$ is equivalent to absolute $(k+1)$-incoherence of $\widetilde{\rho} \in M_{n+1}^{+}$?
\end{itemize}

\noindent \textbf{Acknowledgements.} The authors thank David E.~Speyer for providing the proof of Lemma~\ref{Speyer} \cite{SpeyMO}.  We would also like to thank the authors of \cite{blekherman2022hyperbolic} for bringing our attention to this paper.  N.J.\ was supported by NSERC Discovery Grants RGPIN-2016-04003 and RGPIN-2022-04098. S.M.\ acknowledges hospitality from Mount Allison University during her postdoctoral fellowship. R.P.\ was supported by NSERC Discovery Grant number 400550.  S.P.~was supported by NSERC Discovery Grant number 1174582, the Canada Foundation for Innovation (CFI) grant number 35711, and the Canada Research Chairs (CRC) Program grant number 231250.

\bibliographystyle{alpha}
\bibliography{references}
\end{document}